\documentclass[11pt]{article}
\usepackage{latexsym}
\usepackage{amsfonts,amssymb,amsmath,amsthm}
\usepackage{amstext}
\usepackage{graphicx}
\usepackage{color}
\usepackage{xspace}
\usepackage{url,hyperref}
\usepackage{paralist}

\newtheorem{fact}{Fact}[section]
\newtheorem{lemma}{Lemma}[section]
\newtheorem{theorem}[lemma]{Theorem}

\newtheorem{proposition}[lemma]{Proposition}

\newenvironment{proofof}[1]{\smallskip\noindent{\bf Proof of #1}}%
        {\hspace*{\fill}$\Box$\par}
        
\newcommand{\etal}{et al.\ }

\newcommand{\ex}{{\mathrm E}}

\newcommand{\lpfrac}{\mathrm{LP'_{frac}}}
\newcommand{\lpfracs}{\mathrm{LP^*_{frac}}}
\newcommand{\lpint}{\mathrm{LP'_{int}}}
\newcommand{\lpints}{\mathrm{LP^*_{int}}}

\newcommand{\lpobj}{\mathrm{LP^i_{frac}}}

\newcommand{\cJ}{{\cal J}}

\newcommand{\opt}{\textsc{OPT}{}}
\newcommand{\eps}{\epsilon}

\begin{document}

\title{ Scheduling to Approximate Minimization Objectives on Identical Machines}

\author{Benjamin Moseley\thanks{Tepper School of Business, Carnegie Mellon University and Relational AI. moseleyb@andrew.cmu.edu. Supported in part by a Google Research Award, a Infor Award and NSF Grants CCF-1824303, CCF-1733873 and CCF-1845146. }}

\date{}

\maketitle

\begin{abstract}

This paper considers scheduling on identical machines.  The scheduling objective considered in this paper generalizes most scheduling minimization problems.  In the problem, there are $n$ jobs and each job $j$ is associated with a monotonically increasing function $g_j$. The goal is to design a schedule that minimizes $\sum_{j \in [n]} g_{j}(C_j)$ where $C_j$ is the completion time of job $j$ in the schedule.  An $O(1)$-approximation is known for the single machine case.  On multiple machines, this paper shows that if the scheduler is required to be either non-migratory or non-preemptive then any algorithm has an unbounded approximation ratio.  Using preemption and migration, this paper gives  a $O(\log \log nP)$-approximation on multiple machines, the \emph{first} result on multiple machines.  These results imply the first non-trivial positive results for several special cases of the problem considered, such as throughput minimization and tardiness.  

Natural linear programs known for the problem have a poor integrality gap.  The results are obtained by strengthening a natural linear program for the problem with a set of  covering inequalities we call  \emph{job cover inequalities}. This linear program is rounded to an integral solution by building on quasi-uniform sampling and rounding techniques.

\end{abstract}
\newpage

\section{Introduction}

A common optimization challenge  is scheduling a set of $n$ jobs on $m$ identical machines to optimize the quality of service delivered to the jobs.  The quality of service  objective could be: a delay based objective, such as minimizing the average waiting time;  a fairness objective ensuring resources are shared fairly between jobs, such as the $\ell_2$-norm of the waiting time; or a real-time objective such as ensuring a small number of jobs are not completed by their deadline. 


\medskip
\noindent \textbf{Scheduling Model:} This paper develops  an algorithm that has strong guarantees for most reasonable  objectives.  This work considers the \emph{identical} machines setting where all jobs are available at the same time.  Each job $j$ has a processing time $p_j$.  The job can be processed on $m$ identical machines where the processing time of the job is the same on all machines. This work assumes that \emph{preemption} and \emph{migration} are allowed.  That is jobs can be stopped and resumed at a later time, possibly on a different machine.

This paper initiates the study of the general scheduling problem (GSP) on identical machines.  In this problem, each job $j$ has a function $g_j(t): \mathbb{R}^+\rightarrow \mathbb{R}^+$.  The value of $g_j(t)$ specifies the cost of completing job $j$ at time $t$.  The goal is to design an algorithm that completes each job $j$ at time $C_j$ to minimize $\sum_{j \in [n]} g_j(C_j)$. No assumptions on the functions  are made except that they are positive and non-decreasing, so there never is an incentive to have a job wait longer to be completed.  Note that each job has its own, \emph{individual}, cost function.  In several systems, it is the case that jobs can be associated with distinct cost functions \cite{LoChGo15,MarsTaHu11,MarshallKeFr11}. 

  The problem  generalizes many scheduling objectives. Examples include the following.  In the following descriptions, each job $j$ has a positive weight $w_j$ denoting its priority.

\begin{itemize}
\item \textbf{Weighted Completion Time:} A job's cost is its weight multiplied by its completion time.   The completion time is how long the job waits in the system and this objective focuses on minimizing the priority scaled average waiting time. This objective is captured by setting $g_j(t)= w_j\cdot t$.
\item \textbf{Weighted $k$th Norm of Completion Time:} This objective focuses on minimizing  $\sqrt[k]{\sum_{j \in [n]} w_j C_j^k}$ or, by removing the outer $k$th root, $\sum_{j \in [n]} w_j C_j^k$ . This objective is captured by setting $g_j(t)= w_j\cdot t^k$. This is used to enforce fairness in the schedule and  typically $k \in \{2,3\}$.
\item \textbf{Weighted Throughput Minimization:} The goal is to minimize the weighted number of jobs that miss their deadline.  Each job $j$ has a deadline $d_j$.   Setting $g_j(t) = 0$ for $t \leq d_j$ and $w_j$ otherwise gives this objective.
\item \textbf{Weighted Tardiness:} Each job has no cost if completed before its deadline and otherwise  the job pays its weighted waiting time after is deadline.    Each job  has a deadline $d_j$ and weight $w_j$.  This objective is obtained by setting $g_j(t) = 0$ for $t \leq d_j$ and $w_j(t-d_j)$ otherwise.
\item \textbf{Exponential Completion Time:} In this objective a job's cost grows exponentially with its completion time.  The objective is captured by setting $g_j(t)= w_j\cdot \exp(t)$.
\end{itemize}

These problems have been challenging to understand. The problem considered is NP-Hard, even on a single machine and the cost functions are piecewise linear \cite{HohnJ15}.   It is known that Smith's rule is optimal for minimizing the total weighted completion time  \cite{smith}.     Bansal and Pruhs in a breakthrough result introduced a  $O(1)$-approximation algorithm for the general scheduling problem on a single machine \cite{BansalP14}. Cheung \etal improved this to show a $(4+\eps)$ approximation \cite{CheungMSV17,CheungS11,MestreV14}. Antoniadis \etal gave a quasi-polynomial time approximation scheme on a single machine \cite{AntoniadisHMVW17,HohnMW14}. 

The next step in this line of work is to generalize these techniques to multiple machine environments, but there is a clear barrier when generalizing past approaches to multiple machines. Prior work  introduced a strong linear program that uses a polynomial number of knapsack cover inequalities. See \cite{CarrFLP00} for details on knapsack cover inequalities.  The inequalities in \cite{BansalP14,CheungS11,MestreV14} are  weak in multiple machine environments and result in linear programs with an unbounded integrality gap.

An open question is if there exists  a linear program with a small integrality gap for multiple machines. Further, are there good approximation algorithms for the GSP on multiple machines.

\medskip
\noindent \textbf{Results:} This paper studies the GSP in the identical machine environment. The paper shows the following theorem.  The technical contributions that result in this theorem are the derivation of valid strong linear program inequalities that are used to strengthen a natural linear program relaxation of the problem and an iterative rounding technique that builds on quasi-uniform sampling \cite{Varadarajan10}. 

\begin{theorem}
\label{thm:main}
There is  a randomized algorithm that achieves a $O( \log \log nP)$ approximation in expectation and runs in expected polynomial time for the GSP on multiple identical machines with preemption and migration where $P$ is the ratio of the maximum to minimum job size. 
\end{theorem}

A natural question is if preemption and migration are necessary for an algorithm to have a good approximation ratio.  This paper shows that they are by establishing that any scheduler  required to be non-preemptive or non-migratory  has an unbounded approximation ratio unless $P=NP$.  The proof is deferred to Appendix~\ref{sec:approxhard}.


\begin{theorem}\label{thm:approxhard}
The approximation ratio of any algorithm for GSP is unbounded unless P=NP if either the algorithm is required to be non-migratory or non-preemptive on $m$ identical machines.
\end{theorem}

\medskip
\noindent \textbf{Overview of Technical Contributions:}
The main result is enabled by a set of strengthening inequalities added to a natural linear program (LP) for the problem.  The paper calls these inequalities, \textbf{job cover inequalities}.  See Section~\ref{sec:lp}. These inequalities are needed because without them the LP  introduced in this paper has an unbounded integrality gap even if all jobs arrive at the same time on a single machine\footnote{Without strengthening inequalities and when jobs arrive at the same time on a single machine the LP introduced in this paper can be reduced to an LP used in prior work where the gap is known \cite{BansalP14,CheungS11}}. Other natural LP relaxations, such as a time indexed LP, also have an unbounded gap even on a single machine \cite{BansalP14,CheungS11}. 

Prior work on a single machine also used a set of covering inequalities to strengthen a linear program.   These inequalities consider every interval $I$ and the set of jobs that arrive during the interval $S_I$.  A constraint states that the total processing time of jobs in $S_I$ that are completed after $I$ ends must be greater than the total processing time of jobs in $S_I$ minus the length of $I$ \cite{BansalP14,CheungS11}. This is a covering constraint ensuring that the jobs arriving during $I$ that complete during $I$ have total size at most the length of $I$. These covering constraints are  strengthened using knapsack cover inequalities. If such constraints are satisfied integrally then the Earliest-Deadline-First algorithm can be used to construct a schedule of the same cost as the LP.

A natural idea to extend this to identical machines is to use the same constraint, but the total work completed after $I$ ends must be greater than the size of jobs in $S_I$ minus $m$ times the length of $I$.  This generalization takes into account that each machine can be busy during $I$.   Then the natural next step is to use  knapsack cover inequalities to strengthen this new set of constraints. Unfortunately, it is easy to show such inequalities are insufficient and result in an LP with a large integrality gap.  There are several issues and they are all rooted in the fact that this does not take into account that a job can only be processed on one machine at any point in time.    To overcome this shortcoming, this paper considers covering constraints used to strengthen a minimum cut constraint arises from a natural bipartite flow problem.

This paper proceeds by first reducing the scheduling problem to the problem of finding completion times for each of the jobs, without committing to a schedule.  Once a feasible set of completion times is discovered, the scheduling of jobs can easily be obtained by solving a bipartite flow problem. See Section~\ref{sec:deadline} for details.   Feasible solutions to the bipartite flow problem have a one-to-one correspondence to the original scheduling problem.  We note that the reduction to this flow problem is a well-known scheduling technique. 

 The bipartite graph in the flow problem is used to derive the job cover inequalities. First an LP is written based on the flow problem. To ensure a feasible flow is possible, a set of constraints is added that ensure the minimum cut in the graph is sufficiently large.  Then this set of covering constraints are strengthened.  While these inequalities are used to strengthen constraints that arise from a flow graph, they are different than previously studied flow cover inequalities \cite{FlowCoverBook,FlowCover,LeviLS08}. The key to defining the improved constraints is leveraging  the structure of the minimum cuts in the bipartite graph resulting from the scheduling problem.

The algorithm solves the strong  LP  and rounds the solution.  The idea is to use iterative randomized rounding, but this results in a large approximation ratio.  We remark that standard randomized rounding techniques can be used to obtain a $O(\log n)$ approximation by over sampling variables by a $O(\log n)$ factor and then union bounding over constraints to show they are satisfied with high probability.  However, there are issues with reducing the approximation ratio below this factor; the most challenging is showing that the constraints are satisfied if variables are sampled by a smaller factor, which is what would be needed to reduce the approximation ratio.

 Instead, the algorithm uses an iterative scheme to round the solution. In each iteration, the algorithm  over samples variables by a small factor and modifies the linear program.  The modification ensures that (1) in expectation no variable,  over all iterations, is sampled by more than a $O( \log \log nP)$ factor  than how much it is selected by the optimal LP solution and (2) the relaxation remains feasible.    The scheme builds on techniques of quasi-uniform sampling \cite{ChanGKS12,Varadarajan10} and quasi-uniform iterative rounding \cite{ImM17}.


\medskip
\noindent \textbf{Other Related Work}

\medskip
\noindent
\textbf{Single Machine}. It is known the Smith's rule is optimal for minimizing the total weighted completion time \cite{smith}. The tardiness problem has been challenging to understand.  Lawer \cite{LAWLER} gave a polynomial approximation scheme if jobs have unit size.  Before the work on the GSP in the single machine environment, the previously best known approximation for arbitrary sized jobs was a $(n-1)$-approximation \cite{ChengNYL05}.  For the GSP a $(4+\eps)$-approximation is known \cite{CheungMSV17} and a quasi-polynomial time approximation scheme has been established \cite{AntoniadisHMVW17}.

\medskip
\noindent
\textbf{Identical Machines}. Smith's rule is optimal for minimizing the total weighted completion time \cite{smith}.   It is not difficult to see that there is an optimal algorithm for makespan\footnote{Makespan is equivalent to minimizing the maximum completion.} if jobs can be preempted and migrated across machines.  An PTAS is known if migration is disallowed \cite{HochbaumS85}.  As far as the author is aware, there are no non-trivial results known for exponential completion time and tardiness on multiple machines.

\section{Preliminaries}
\label{sec:prelim}

In the general scheduling problem (GSP), there is a set $J$ of $n$ jobs.  Each job $i$ has integer processing time $p_i$. Let $P$ denote the  maximum job size and assume the minumum job size is one.   The jobs are to be scheduled on a set $M$ of $m$ identical machines that can schedule one job at any point in time. It is assumed that time is slotted and job $i$ must be scheduled for $p_i$ time units over all machines.  Jobs can be preempted and migrated across machines.  Jobs cannot be scheduled on more than one machine simultaneously.  Every job $i$ is associated with a function $g_i: \; \mathbb{R}^+ \rightarrow \mathbb{R}^+$ where $g_i(t)$ specifies the cost of completing job $i$ at time $t$. Without loss of generality, assume that $g_i(0) = 0$ for all jobs $i$.  The only assumption on $g_i(t)$ is that it is a non-negative non-decreasing function. Under a given schedule, job $i$ is completed at time $C_i$. The goal is for the scheduler to minimize $\sum_{i\in J} g_i(C_i)$.   Note that it can be assumed that all jobs are completed by time $nP$.

\section{Scheduling Jobs with Deadlines}
\label{sec:deadline}

This section shows that if the completion times of the jobs are fixed then there is a method to determine how to schedule the jobs at or before their completion times or determine if such a schedule is not possible.  Notice that if such a schedule is feasible then this ensures the objective is either the same or smaller in the computed schedule than if all jobs are completed at exactly their given completion times.  Let job $i$ have a given completion time $C_i$. The completion time $C_i$  is interpreted as job $i$'s deadline.     

The method to construct a schedule for the jobs is to setup a flow problem. Setting up a flow graph to determine if a set of jobs can be feasibly scheduled is a standard scheduling technique (e.g. \cite{ChuzhoyGKN04}), but is presented here so that later this graph can be used in a linear program formulation. Consider creating a bipartite flow graph $G= (\{ s,d\} \cup A \cup B,E)$ where $A$ contains a node $a_i$ for every job $i$ and $B$ contains a node $b_t$ for every time step $t$. There is additionally a source $s$ with an outgoing edge to each node $a_i$ in $A$ with capacity $p_i$.  There is a sink node $d$ that has an incoming edge from each node in $B$ with capacity $m$.  Finally there is an edge from $a_i \in A$ to $b_t \in B$  of capacity $1$ if and only if $t \leq C_i$.  See Figure \ref{fig:flow}.

\begin{figure}{r}
\vspace{-.8cm}
\begin{center}
\includegraphics[width=.3\textwidth]{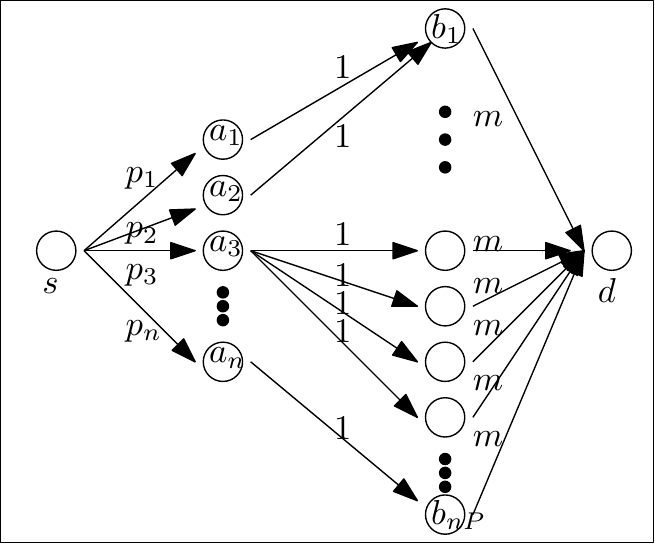}
\caption{The graph $G$.    A job $i$ needs to assign $p_i$ units of flow (processing time) to machines before $C_i$. Machines can process up to $m$ units at each time.  }\label{fig:flow}
\end{center}
\vspace{-1.cm}
\end{figure}

A  set of completion times are feasible if and only if there is a feasible flow in this network of value $\sum_{i \in [n]} p_i$.   This is because a job can be scheduled for a unit at each time during $[0,C_i]$ and must be scheduled for $p_i$ units total.  Further, every time step can schedule up to $m$ jobs. 
 
  There are two messages to takeaway from this. One is that the problem can be solved by only knowing completion times for the jobs. The other is that this flow graph can be used to determine if a set of completion times can be associated with a valid schedule.  A set of completion times are said to be \textbf{valid} if there is a schedule that completes each job only earlier than the given completion time.
 
 The following theorem follows from the construction of $G$.
 
 \begin{theorem}
 \label{thm:flow}
 A set of completion times is valid if and only if there is a feasible maximum flow of value $\sum_{i=1}^n p_i$ in the flow graph $G$.
 \end{theorem}
 
 Given a set of valid completion times, one can construct a feasible schedule using the flow graph. That is, an assignment of jobs to machine at each time step.  Unfortunately, the graph has size $\Omega(nP)$ and could be exponential in size. Recall that $P$ is the ratio of the maximum to minimum job size. There is a polynomial time algorithm that constructs a feasible schedule in time polynomial in $n$ and $\log W$ given a set of valid completion times.    Here $W$ denotes the maximum value of $g_j(t)$ for $t \leq nP$.   This algorithm is omitted due to space. It can be obtained by rounding possible completion times to geometrically increasing times.


\section{Strengthened Linear Program with Job Cover Inequalities}
\label{sec:lp}

Consider the following natural integer program.  The constraints in the program come from the flow graph of the prior section.  Let $x_{j,t}$ be $1$ if job $j$ is not completed at time $t$ and $0$ otherwise.\footnote{Note that this is \emph{not} the standard time indexed LP where the variables represent the amount $j$ is processed at time $t$.}   This implies that $x_{j,t}$ is continuously $1$ for $t$ less than $j$'s completion time and then $0$ for times $t$ thereafter in an integer solution.   Let $T$ be the set of all time slots and $J$ the set of all jobs. By assuming $g_{j}(0) = 0$ for all $j$, the objective function is a telescoping summation for each job $j$ whose value will be $g_j(t)$ if $t$ is the last time $j$ is not fully processed. The first set of constraints says that if job $j$ is completed at time $t-1$ then it is also completed at time $t$. The second set of constraints are more involved.  One can think of the latest time $t$ that $x_{j,t}=1$ as being the completion time of $j$.  Given this, the constraint says that every cut in the flow network from Section~\ref{sec:deadline} has value at least $\sum_{i \in J} p_i$.  This is a valid constraint by Theorem~\ref{thm:flow} and the maximum-flow minimum-cut theorem.  Note that $J'$ corresponds to jobs whose nodes are on the side of the cut with the source $s$.  Similarly, $T'$ corresponds to time steps whose nodes are on the side of the cut with the sink $d$.  

\vspace{-.5cm}
\begin{small}
\begin{align}
\min  \sum_{j \in J} \sum_{t \in T}  &x_{j,t} (g_j(t) - g_j(t-1))  \label{IP:first} \\
\mbox{s.t. } \;\;\;\; x_{j,t} &\leq x_{j,t-1} & \forall j \in J , t \in T  \label{constraint:greater}\\
\sum_{j \in J \setminus J'}p_j + \sum_{j \in J'} \sum_{t \in  T'} x_{j,t} + \sum_{t \in T\setminus T'} m &\geq \sum_{j \in J} p_j &\forall J' \subseteq J,T' \subseteq T \label{constraint:weak}\\
x_{j,t} \in \{0,1\} & &\forall j\in J,t \in T \nonumber
\end{align}
\end{small}
\vspace{-.8cm}

   The goal is to derive a set of valid strengthening inequalities for the IP (\ref{IP:first}). These inequalities are used to  strengthen the minimum cut constraints. These inequalities are needed because without them the LP has an unbounded integrality gap even if all jobs arrive at the same time on a single machine \cite{CheungS11}. Other natural LP relaxations, such as a time indexed LP also have an unbounded gap even on a single machine. 

We can derive a set of strengthening inequalities for this linear program that replace the set of constraints (\ref{constraint:weak}).   The proof establishing validity of the following constraintsin Appendix \ref{sec:derivation} .  The strengthening focuses on the constraints in  (\ref{constraint:weak}) for $J' = J$.  In the end, the derived constraints are strictly stronger than the above and one need not consider the other sets $J'$.   

Fix  $J' = J$ and consider the  constraints $\sum_{j \in J} \sum_{t \in  T'} x_{j,t} + \sum_{t \in T\setminus T'} m \geq \sum_{j \in J} p_j$ for all $T' \subseteq T$. For a job $j$ and a collection of time steps $T'$ let $E(T',j)$ be the set of up to $\min\{p_j, \sum_{i \in J} p_i - m | T\setminus T'| \}$ \textbf{earliest} time steps in $T'$.  The full analysis first shows that we can strengthen this to $$\sum_{j \in J} \sum_{t \in  E(j,T')} x_{j,t} + \sum_{t \in T\setminus T'} m \geq \sum_{j \in J} p_j \mbox{ for all } T' \subseteq T.$$  Notice the first summation now only considers time steps in $E(T',j)$ for each job $j$.

The proof further improves these inequalities by taking inspiration from knapsack cover inequalities.  Let $D$ be a \textbf{vector} where $D_j$ is a time corresponding to job $j$.   Intuitively, $D_j$ is a lower bound on the completion time for  job $j$. The constraints below say that even if all jobs $j$ are set to have completion times \textbf{at least} $D_j$  then the constraints should still be satisfied. 

Fix $T' \subseteq T$ and a vector $D$ of completion times for every job.  Consider the constraint for $T'$ in the above set of inequalities.  If job $j$ is given a completion time of at least $D_j$ then  $j$ will contribute at least $\sum_{t \in  [0, D_j] \cap E(T',j) }1$ to the left side of the inequality.  Let $V(T',D) = \sum_{i\in J } p_i - \sum_{t \in T \setminus T'} m - \sum_{j \in J}\sum_{t \in  [0, D_j] \cap E(T',j) }1$.    Let $E(T',D,j)$ be the \textbf{earliest}  $V(T',D)$ time steps  in $E(T',j)$ later than $D_j$.  The new constraints are as follows.  These are the \textbf{job cover inequalities}.

\begin{eqnarray}
\sum_{j \in J} \sum_{t \in E(T',D,j)} x_{j,t} \geq V(T',D)  \;\;\;\;  \forall T' \subseteq T, \forall D, V(T',D)  >0 \label{constraint:nospeed1}
\end{eqnarray}

The validity of these inequalities is not difficult to show, but technical. Notice that the number of inequalities is large.  We can somewhat reduce the number of constraints as follows.  It can be established that any \emph{integer} solution satisfies all the constraints in (\ref{constraint:nospeed1}) if and only if the following are satisfied. The proof is in Appendix \ref{sec:continuousT}.


\begin{eqnarray}
\sum_{j \in J} \sum_{t \in E(T',D,j)} x_{j,t} \geq V(T',D)  \;\;\;\;   \forall b \in [0,\infty], T' =[b,\infty], \forall D, V(T',D)  >0\label{constraint:nospeed}
\end{eqnarray}

 These constraints state that if the subset of constraints in  (\ref{constraint:nospeed1}) are satisfied for any $D$ and all sets $T'$  that consist of a continuous set of time steps from some  time $b$ to time $\infty$ then all of the constraints in  (\ref{constraint:nospeed1})  are satisfied (for any $T'$).  Due to this, we will only need to use the constraints in (\ref{constraint:nospeed}) that restricts the sets $T'$.

The constraints (\ref{constraint:weak}) in the  IP are replaced by the constraints in (\ref{constraint:nospeed}).   Throughout the paper these constraints are discussed and it is said that a fixed constraint is defined by the set $T' = [b,\infty]$ and a vector $D$.

\medskip \noindent \textbf{Note on Solving  the LP:}  The IP is relaxed to a LP.  The LP is solved and then subsequently rounded to an integer solution.  There are an exponential number of constraints.  To solve the LP, the ellipsoid method is used.  The author does not know of an efficient separation oracle for the set of constraints in (\ref{constraint:nospeed1}).   The reduced set of constraints in (\ref{constraint:nospeed}) are the only constraints needed for the analysis.  For this set of constraints, an efficient dynamic programming algorithm can be used as a separation oracle.  The separation oracle can be found in the appendix.

\section{The Rounding Algorithm}
\label{sec:alg}

In this section the algorithm for rounding a fractional LP solution to an integral solution is described.   Recall in Section~\ref{sec:deadline} it was shown how to assign jobs to machines and time slots if a valid set of completion times have been established.   The remaining goal is to design an algorithm that rounds a fractional LP solution to an integral solution, giving the completion times for the jobs.  

The algorithm  takes as input a fractional solution to the LP $x'$.  The input solution is for the LP given in the previous section with constraints (\ref{constraint:weak}) replaced with the derived constraints (\ref{constraint:nospeed}).  The algorithm rounds the $x'$ solution to an \emph{integral} solution $x^*$.  If this solution is feasible, then the algorithm terminates.  Otherwise  $x^*$ is modified to obtain a new feasible fractional solution  $\tilde{x}^*$ that the algorithm recurses on. 

\medskip
\noindent \textbf{Informal Algorithm Description and Intuition:} The algorithm runs in phases. During a phase the algorithm finds a completion time for each job.  The completion times are pushed back to later times in each phase.  In a fixed phase the algorithm runs a randomized rounding procedure to sample completion times for the jobs.  Roughly, each completion time will be over sampled by a $\Theta(\log \log nP)$ factor over the fractional part of the LP.  After sampling an integer solution $x^*$ is created and variables are set corresponding to the sampled completion times of the jobs.   Some of the constraints in the LP will be satisfied.  These constraints will always remain satisfied because each job's completion time is only pushed back to a later time in subsequent iterations. The algorithm needs to satisfy the remaining constraints, which it does by recursing on a fractional LP solution $\tilde{x}^*$ to make completion times later. The fractional solution $\tilde{x}^*$ is constructed by increasing variables in the integral solution $x^*$.

In each iteration, the cost of the sampling is bounded by a $\Theta(\log \log nP)$  factor more than the \emph{fractional} portion of the linear program objective in expectation.   Ideally, one can use standard iterative rounding for the recursion.  However, naive approaches could have large cost as it is difficult to bound the number of times the algorithm recurses and therefore difficult to bound how many times a variable in the LP is sampled.  A standard approach will result in a $\Theta(\log nP)$ approximation.

The idea is to only include some fractional variables  in the linear program solution that the algorithm recurses on.  The variables that are included can be slightly larger than their original value.   The essential properties are that (1)  the linear program is feasible and (2) each fractional variable is set to $0$ with constant probability.  Using (2) it will be shown that  the expected value of each variable drops by a constant (e.g. $\frac{1}{2}$) factor of its value at the beginning of the iteration. If this is established, then the probability a completion time is sampled decreases geometrically over the phases and we can bound the algorithm's objective  by the cost of the sampling in the first phase, a $\Theta(\log \log nP)$ factor within the original LP object in expectation.  

The recursion will determine fractional variables to set to satisfy all constraints. One can think of the fractional variables as  \textbf{fractional  completion times} for the jobs. The idea is to associate each constraint with a set of fractional completion times that are \emph{critical} for satisfying the constraint in the solution $x'$. This will not be all completion times used to satisfy the constraint, but an essential subset of them. By slightly increasing the linear program variables in the integral solution $x^*$ to get a solution $\tilde{x}^*$ it is the case that only the critical completion times are needed to satisfy their corresponding unsatisfied constraints. 

If a constraint is unsatisfied, then the solution $\tilde{x}^*$ will set positive fractional values for all completion times critical for this constraint to satisfy it. The analysis will show that each fractional completion time is  $0$ (not increased) in $\tilde{x}^*$  with constant probability. For a completion time to be  $0$ in the recursion we need to ensure all constraints are satisfied where the completion time is critical by the integer solution $x^*$.   Unfortunately, a completion time could be critical for many constraints. Due to this, it is insufficient to show each constraint is satisfied with good probability and then union bound over all constraints. 

Fix a fractional completion time. The proof establishes that with good probability \emph{every} constraint that the completion time is critical for is satisfied in $x^*$.  This ensures that the the completion is $0$ with constant probability.  This will be used to show that the expected value of a fractional variable in the LP decreases geometrically over the iterations.  A completion time is over-sampled by at most a $\Theta(\log \log nP)$ factor as compared to the original LP solution in expectation over all iterations. The cost is as if only one iteration of uniform sampling occurred. This is similar to the analysis approach used in quasi-uniform sampling \cite{Varadarajan10} and rounding \cite{ImM17}.


\subsection{Formal Algorithm Description} The $i$th phase of the algorithm is the following.  The algorithm  recurses on the following until all constraints are satisfied.  Let the input to the first phase be $x$, the optimal fractional solution to the LP.  The algorithm utilizes randomized rounding parameterized by $c \leq 1$. The value of $c$ will be set to be $\frac{1}{\Theta(\log \log nP)}$.

\medskip
\noindent \textbf{Phase $i$ of the Algorithm:}   The algorithm first uses \textbf{randomized rounding}. Let $x'$ be a feasible fractional solution to the LP at $i$th phase of the algorithm.       For each job $j$, the algorithm chooses a value $\alpha_j \in [0,1]$ uniformly at random and independently.  Let $C_{j,\alpha}$ be the latest time $t$ where $x'_{j,t } \geq c \alpha_j$.  Let $\beta'_j$ be the latest time $t$ where $x'_{j,t} = 1$.  Let $\lpint = \sum_{j} g_j(\beta'_j)$ be the total \textbf{integral cost} of the LP solution $x'$ and let $\lpfrac = \sum_{j\in J} \sum_{t > \beta'_j} x'_{j,t}(g_j(t) - g_j(t-1))$ be the total \textbf{fractional cost} of the LP solution $x'$.
  
The algorithm modifies the LP solution $x'$ to get a new  (possibly infeasible) solution $x^*$. This is further modified to get a feasible solution  $\tilde{x}^*$  that the algorithm  recurses on.  The algorithm sets $x^*_{j,t} =1$ for all $j$ and $t$ where $t \leq C_{j,\alpha}$ and $0$ otherwise.  If all constraints are satisfied in the LP, then the algorithm sets $C_j^* = C_{j,\alpha}$ and returns this set of completion times as the final solution.  If not, then the algorithm further modifies $x^*$ as described below and recurses on phase $i+1$.  Let $x^*$ denote the current integral solution and $\tilde{x}^*$ a fractional solution resulting from the following modification to $x^*$.


Consider any constraint in (\ref{constraint:nospeed}) defined by a set of time steps $T'$ and a vector of completion times $D'$ that is not satisfied by $x^*$. Assume $D'$ is chosen so that $D'_j \geq C_{j,\alpha}$ for all $j$.  We only need to consider  these constraints because $x_{j,t} =1$ for $t\leq  C_{j,\alpha}$.  Recall that we may assume $T'$ contains a continuous set of time steps beginning with some time $t_{T'}$ and going to $\infty$.   The algorithm  identifies a set of pairs of jobs and completion times that are fractionally chosen in $x'$ that are `critical' for satisfying this constraint.  Intuitively, we will need to include these same fractional completion times in $\tilde{x}^*$. 


 Let $J_{T', D'}$ be the set of jobs $j$ where  $\sum_{t \in E(T',D',j)} x'_{j,t} >0$. These are jobs used to satisfy the constraint in $x'$. The following two sets of jobs can be thought of as \emph{not} critical for the constraint.

\begin{enumerate}
\item Order the jobs $j$ in $J_{T',D'}$  as $1, 2, 3, \ldots  |J_{T',D'}|$ in decreasing order of $D'_j$.  Let $M_{T',D'} \subseteq J_{T',D'}$ be the smallest prefix of jobs $1, 2, \ldots k$ from this order such that $\sum_{j=1}^k  \sum_{t \in E(T',D',j)} x'_{j,t} \geq \frac{1}{10} V(T',D')$. 
\item Order the jobs $j$ in $J_{T',D'}$ as $1, 2, 3, \ldots, |J_{T',D'}|$ in increasing order of $D'_j$.  Let $L_{T',D'} \subseteq J_{T',D'}$ be the smallest prefix of jobs $1, 2, \ldots k$ in this order such that $\sum_{j=1}^k  \sum_{t \in E(T',D',j)} x'_{j,t} \geq \frac{1}{10} V(T',D')$.  
\end{enumerate} 

We will say that  the  jobs in $\mathcal{J}_{T',D'} = J_{T',D'} \setminus  (L_{T',D'} \cup M_{T',D'})$ are \textbf{critical} for satisfying the constraint $T',D'$. Let $\mathcal{J}^* = \{j \; | \; \exists T'D' \mbox{ s.t.  the constraint for $T'$ and $D'$ is unsatisfied in $x^*$ and $j \in \mathcal{J}_{T',D'}$}  \}$.  This is the set of all critical jobs for constraints that are not satisfied by $x^*$. 

 For each job $j \in \mathcal{J}^*$ set $\tilde{x}^*_{j,t'}$ to be  $10x'_{j,t'}$ for all $t' \geq C_{j,\alpha} $.  Note that if $c < \frac{1}{10}$ and $C_{j,\alpha} \leq t'$ it is the case that $\tilde{x}^*_{j,t'} \leq c \leq 1/10$, so $\tilde{x}^*$ is in $[0,1]$ as desired. After performing all updates, recurse.

\medskip
\noindent \textbf{Terminating Condition:}  The above description states that the algorithm terminates when all constraints are satisfied by the integer solution $x^*$.   The analysis will show that this occurs in polynomial time in expectation.

\section{Bounding the Cost and Feasibility of the Algorithm}
\label{sec:costbound}

In this section, the correctness of the algorithm is established and  the total cost of the algorithm  is bounded as well as the running time. The analysis has several intermediate goals.  One is to show that the resulting solution $\tilde{x}^*$ is feasible. Another is to show that the increase in cost of the randomized rounding is bounded by the cost of the fractional part of the LP in one iteration.   The final goal is showing that the expected value of the objective for the fractional part of the LP solution decreases significantly in each iteration.  After establishing these facts Theorem~\ref{thm:flow} will prove the main theorem. 

\medskip
\noindent \textbf{Feasibility of the Algorithm:}\label{sec:feasible} We now establish that $\tilde{x}^*$ is feasible for the LP.  This ensures the algorithm constructs a feasible solution.   The proof is omitted due to space and is in Appendix~\ref{sec:omitted}. The proof follows  by the fact that each unsatisfied constraint is associated with variables in the solution $x'$ which are within a constraint multiplicative factor of satisfying the constraint.   In the recursion, these variables are included in $\tilde{x}^*$ with their fractional valuesincreased by a factor $10$ over $x'$ ensuring their corresponding constraint is satisfied. 

\begin{lemma}
\label{lem:feasible}
In any iteration of the algorithm, if $x'$ is a feasible LP solution  then the algorithm constructs a feasible solution $\tilde{x}^*$ at the end of the iteration for any $c \leq 1/10$.  
\end{lemma}

\medskip
\noindent \textbf{Bounding the Cost and Running Time of the Algorithm:} This section bounds the cost of the LP solution $\tilde{x}^*$ and the running time of the algorithm.   First the cost of the randomized rounding is bounded.  This bounds the cost of the \emph{intermediate} integral solution $x^*$. The following lemma shows that the expected cost \emph{increase} of the LP solution $x^*$ over $x'$ is bounded by $\frac{1}{c}\lpfrac$.    Recall that  $x^*$ is the solution obtained by only the randomized rounding part of the algorithm and it is a possibly infeasible integer solution. After this lemma, the cost of the solution  $\tilde{x}^*$ is bounded.  This is the solution the algorithm recurses on.  Combining the cost over the entire algorithm is bounded.

 The following lemma bounds the cost of the solution $x^*$.  The proof can be found in Appendix~\ref{sec:omitted}.  The proof follows by a standard analysis of randomized rounding.

\begin{lemma}
\label{lem:expectedcost}
The total expected different in cost  of $x^*$ and $x'$ is at most  $\frac{1}{c}\lpfrac$.  That is, $\ex[ \sum_{j \in J} \sum_t (x^*_{j,t} - x'_{j,t}) (g_j(t) - g_j(t-1))]  \leq \frac{1}{c} \sum_{j \in J} \sum_{t } \lpfrac$.
\end{lemma}

Let $\lpfracs := \sum_{j\in J} \sum_{t > C_{j,\alpha}} \tilde{x}^*_{j,t}(g_j(t) - g_j(t-1))$ be the fractional cost of $\tilde{x}^*$ and $\lpints =\sum_{j\in J}  g(C_{j,\alpha})$ be the integral cost.  Note that $\lpints$ is precisely the objective of the integral solution $x^*$. The key to bounding the cost of the algorithm is to show that the fractional cost of the LP solution decreases by a constant factor in each iteration. This is stated in the following lemma.  The proof is deferred due to space.  This lemma is the most interesting part of the analysis and the proof is presented in Section~\ref{sec:type1}.

\begin{lemma}
\label{lem:costtotal}
In any iteration of the algorithm, $\ex[\lpfracs] \leq \frac{1}{4}\lpfrac$. 
\end{lemma}

Using this lemma, the total cost of the algorithm can be bounded.  

\begin{lemma}
\label{lem:Finalcosts}
Let $\opt$ be the optimal feasible objective to the LP.  It is the case that the algorithm's total cost is at most $\frac{2}{c} \opt$ in expectation.
\end{lemma}
\begin{proof}

Let $\lpobj$ denote the fractional part of the objective in the LP solution at the beginning of the $i$th iteration of the algorithm.  Note that $\mathrm{LP^1_{frac}}\leq \opt$.  Inductively,   Lemma~\ref{lem:costtotal} gives that $\ex[\lpobj] \leq \frac{1}{4^i} \opt$.  

Lemma~\ref{lem:expectedcost} ensures that the integral portion of the LP objective increases by at most $\frac{1}{c} \lpobj$ in each iteration.  Thus the total cost can be bounded by $\frac{1}{c}\sum_{i=1}^\infty \frac{1}{4^i} \opt \leq \frac{2}{c} \opt.$
\end{proof}

The following proposition bounds the run time of the algorithm and the proof is in Appendix~\ref{sec:omitted}. This follows because the previous lemma will ensure the variables in the LP converge to $0$ after $O(\log nP)$ iterations.

\begin{proposition}
\label{prop:runtime}
The algorithm runs in polynomial time in expectation.
\end{proposition}

 Lemma~\ref{lem:Finalcosts} bounds the objective of the algorithm.  Lemma~\ref{lem:feasible} ensures that the algorithm constructs a feasible solution.  Finally, Propostion~\ref{prop:runtime} shows the algorithm runs in polynomial time.  Together, this proves the main result, Theorem~\ref{thm:main}.

\subsection{Proof of Lemma~\ref{lem:costtotal}: Expected Decrease in the Fractional Objective}
\label{sec:type1}

This section proves Lemma~\ref{lem:costtotal} for a fixed iteration of the algorithm.  Fix a job $j$ and the fractional solution $x'$ that is input to the rounding algorithm in this iteration.   The goal  is to show that the probability that $j$ is in $\mathcal{J}^*$ is small and therefore the algorithm only includes integral variables $\tilde{x}^*_{j,t}$  for job $j$ when it recurses.   In particular, the goal is to show the following lemma.  In the following, $n$ and $P$ are assumed to be sufficiently large.

\begin{lemma}
\label{lem:t1prob}
Fix any job $j^*$.  With probability at most $\frac{2}{\log^2 nP} \leq 1/40  $ it is the case that there is a  constraint $T',D'$ unsatisfied by $x^*$ and $j^* \in \mathcal{J}_{T',D'}$ when $c \leq \frac{1}{1000\log \log nP}$.  
\end{lemma}

This lemma implies Lemma~\ref{lem:costtotal}.  

\begin{proofof}[Lemma~\ref{lem:costtotal}] Fix any job $j$ and an iteration of the algorithm.   If there is a constraint $T',D'$ unsatisfied, $j \in  \mathcal{J}_{T',D'}$, $x'_{j,t}$ is fractional, $t \geq C_{j,\alpha}$ and $t \in E(T',D',j)$  then the algorithm sets $\tilde{x}^*_{j,t}$ to be $10x'_{j,t}$. 

 This event increases the cost of $\tilde{x}^*$ by at most $10\sum_{t > C_{j,\alpha}} x'_{j,t}(g_j(t) - g_j(t-1))$ and it happens for some $T',D'$ with probability at most $1/40$.  The total expected cost $\lpfracs$ is at most the following.    
  
  \vspace{-.3cm}
\begin{eqnarray*}
&&10\sum_{j \in J}\sum_{t > C_{j,\alpha}} x'_{j,t}(g_j(t) - g_j(t-1)) \Pr[\exists T',D' | \; j \in \mathcal{J}_{T',D'} \mbox{ and constraint  } T', D' \mbox{ unsatisfied by $x^*$} ]\\
&&\leq 10\sum_{j \in J} \sum_{t > C_{j,\alpha}} x'_{j,t}(g_j(t) - g_j(t-1)) \frac{1}{40} \leq \frac{1}{4}\sum_{j \in J}  \sum_{t > C_{j,\alpha}} x'_{j,t}(g_j(t) - g_j(t-1))  \;\;\;\; \mbox{[Lemma~\ref{lem:t1prob}]}
\end{eqnarray*}

Thus, the expected cost of $\lpfracs$ decreases by a factor $1/4$ over $\lpfrac$. 
\end{proofof}

The remaining goal of this section is to prove Lemma~\ref{lem:t1prob}.   The proof begins by observing that since the solution $x^*$ is integral if a constraint $T'$ and $D'$ is satisfied for a particular set $D'$ then the constraints for $T'$ and all sets $D''$ are satisfied.   This will allow us to focus on constraints for one special set $D'$. The proof can be found in Appendix~\ref{sec:omitted}.

\begin{proposition}
\label{prop:allD}
Let $t_{T'}$ be some time step and $T' = [t_{T'},\infty]$. Let $D'$ be set such that $D'_j$ is the latest time $t$ where $x'_{j,t} \geq c$ for all jobs $j$.   If the constraint for $T'$ and $D'$ is satisfied in the solution $x^*$ then all constraints for $T'$ and any set $D''$ are satisfied in the soultion $x^*$.
\end{proposition}

For the remainder of the proof, fix $D'$ to be as described in the prior lemma. Now we establish some basic propositions on which jobs contribute to a constraint.  This will be useful for identifying critical jobs.      Note that the two propositions below are different depending on the ordering of the times considered. The proofs are omitted and can be found in Appendix~\ref{sec:omitted}.

\begin{proposition}
\label{prop:during}
Fix any job $j$ and two sets $T = [t_{T}, \infty]$ and  $T' = [t_{T'}, \infty]$ where $D'_j \leq t_{T} < t_{T'}$. For any fractional LP solution $x$ satisfying constraints (\ref{constraint:greater})  if $V(T,D') \geq \frac{1}{2} V(T',D')$ then it is the case that $\sum_{t \in E(T',D',j)} x_{j,t} \leq 2\sum_{t \in E(T,D',j)} x_{j,t}$. 
\end{proposition}

\begin{proposition}
\label{prop:before}
Fix any  job $j$ and two sets $T = [t_{T}, \infty]$ and  $T' = [t_{T'}, \infty]$ and $D'_j > t_{T} > t_{T'}$. For any fractional LP solution $x$ satisfying constraints (\ref{constraint:greater})  if  $V(T,D') \geq \frac{1}{2} V(T',D')$ then  $\sum_{t \in E(T',D',j)} x_{j,t} \leq 2\sum_{t \in E(T,D',j)} x_{j,t}$.
\end{proposition}

The next lemma establishes that for any fixed constraint $T'$ and $D'$, it is the case that the constraint is satisfied with good probability in the integral solution $x^*$. Further, the constraint is satisfied if only the jobs in $M_{T',D'}$ (or $L_{T',D'}$) are considered in the summation in the left hand side of the constraint.  
Due to space, the proof can be found in Appendix~\ref{sec:omitted}.

\begin{lemma}
\label{lem:offlinewhp}
Fix any $T' = [t_{T'},\infty]$ and let the vector $D'$ contain the time $D'_j$ that is  the latest time $t$ where $x'_{j,t} \geq c$ for all jobs $j$. With probability at least $1-\frac{1}{ \log^{10} nP}$ it is the case that $\sum_{j \in L_{T',D"}} \sum_{t \in E(T',D',j)} x^*_{j,t} \geq 10V(T',D')$ and $\sum_{j \in M_{T',D'}} \sum_{t \in E(T',D',j)} x^*_{j,t} \geq 10V(T',D')$ when $c \leq \frac{1}{1000\log\log nP}$.
   \end{lemma}


Fix any job $j^*$. Group constraints into two classes for job $j^*$. The first class are those constraints $T'=[t_{T'},\infty]$ where $t_{T'} > D'_{j^*}$ and the second class are the remaining constraints.  It will be shown separately for both groups of constraints that they are all unsatisfied with small probability.

 In the following lemma, the first class of constraints are considered.  This lemma heavily relies on Proposition~\ref{prop:during}. 


\begin{lemma}
\label{lem:afterd}
Fix any job $j^*$.  With probability at most $1/\log^2 nP$ it is the case that there exists a  constraint $T',D''$ unsatisfied by $x^*$, $j^* \in \mathcal{J}_{T',D''}$ and  $t_{T'} > D'_{j^*}$  when $c \leq \frac{1}{1000\log \log nP}$.  
\end{lemma}
\begin{proof}

Fix any job $j^*$.     Let $D'$ be set such that $D'_{j}$ is the latest time $t$ where $x'_{j,t} \geq c$ for all jobs $j$. The proof will establish that with probability greater $1-\frac{1}{\log^2 nP}$ it is the case that the constraints for $D'$ and any $T'=[t_{T'} ,\infty]$ with $t_{T'} > D'_{j^*}$ are satisfied by $x^*$.  Applying Proposition~\ref{prop:allD} this implies that $x^*$ satisfies the same allowing for any constraint $D''$, proving the lemma.

Geometrically group constraints based on the value of $V(T',D')$.  Let $\mathcal{C}_k$ contain the set $T' = [t_{T'} ,\infty]$ if  $j^* \in \mathcal{J}_{T',D'}$,  $2^k \leq V(T',D') < 2^{k+1}$ and $t_{T'} > D'_{j^*}$ for any integer $0 \leq k \leq \log nP$.   

Fix $k$ and the set $T' \in \mathcal{C}_k$ such that $t_{T'}$ is as late as possible.  Let $L_{T',D'}$ be as described in the algorithm definition. We will establish that if  $\sum_{j \in  L_{T',D'}} \sum_{t \in E(T',D',j)} x^*_{j,t} \geq 10V(T',D')$  then all constraints $V(T'',D')$ for any $T'' \in  \mathcal{C}_k$ are satisfied.  Once this is established, this will complete the proof as follows.  We apply  Lemma~\ref{lem:offlinewhp} stating that $\sum_{j \in  L_{T',D'}} \sum_{t \in E(T',D',j)} x^*_{j,t} \geq 10V(T',D')$ occurs with probability at least $1-\frac{1}{ \log^{10} nP}$.  By union bounding for all $\log nP$ values for $k$ the lemma follows.

Say that  $\sum_{j \in  L_{T',D'}} \sum_{t \in E(T',D',j)} x^*_{j,t} \geq 10V(T',D')$.  Consider any set $T'' \in \mathcal{C}_k$.   By definition of the  set  $L_{T',D'}$ it is the case that $D'_j \leq D'_{j^*}$ for all $j \in L_{T',D'}$.  Hence, $D'_{j} \leq D'_{j^*} \leq t_{T''} \leq t_{T'}$. Thus, Proposition~\ref{prop:during} and the geometric grouping of constraints gives that $\sum_{j \in  L_{T',D'}} \sum_{t \in E(T'',D',j)} x^*_{j,t} \geq \frac{1}{2} \sum_{j \in  L_{T',D'}} \sum_{t \in E(T',D',j)} x^*_{j,t} \geq 5V(T',D') \geq V(T'',D')$.    Thus the constraint for $T''$ and $D'$ is satisfied, proving the lemma.
\end{proof}

Similar to the previous lemma, in the following lemma it is shown that all of the second class of constraints are satisfied with good probability.  This lemma heavily relies on Proposition~\ref{prop:before}. This proof is similar to the prior lemma and can be found in Appendix~\ref{sec:omitted}.

\begin{lemma}
\label{lem:befored}
Fix any job $j^*$.  With probability at most $1/\log^2 nP$ it is the case that there exists a  constraint $T',D''$ unsatisfied by $x^*$, $j^* \in \mathcal{J}_{T',D'}$ and  $t_{T'} < D'_{j^*}$  when $c \leq \frac{1}{1000\log \log nP}$.  
\end{lemma}

For sufficiently large $n$ and $P$, a union bound and Lemmas~\ref{lem:afterd} and \ref{lem:befored} prove Lemma~\ref{lem:t1prob}.

\section{Conclusion}

This paper introduced a new set of strong inequalities for scheduling problems on multiple identical machines.  Using these inequalities, the paper showed an iterative algorithm that rounds a fractional LP solution to an integral solution which achieves an $O(\log \log nP)$ approximation for most reasonable scheduling minimization problems.  

An open question is if there are algorithms with an $O(1)$ approximation ratio for GSP on identical machines.  It is also of interest to determine if the inequalities introduced can be extended to other bipartite assignment problems.   Can these inequalities  be extended to more general environments, such as the related machines  or restricted assignment settings?  Could a similar analysis be used when jobs arrive over time?  These cases introduce new technical hurdles, but $O(1)$ approximation algorithms could be possible  by leveraging the given constraints and assuming preemption and migration are allowed.\footnote{We remind the reader that any algorithm for GSP on identical machines has an unbounded approximation ratio if either preemption or migration are not allowed}  For example, if jobs arrive over time then  a generalization of the LP with the strengthened constraints can be derived. However, the rounding becomes challenging because there appears to be no reduction showing only a polynomial number of time sets $T'$ need to be considered in the set of constraints, like  was established in this paper.  Due to this, it is challenging to show all exponential number of constraints based on sets of times are satisfied by the rounding algorithm. 

\bibliography{refsGenCost}

\appendix

\section{Strengthening the Linear Program}
\label{sec:derivation}

The goal is to derive a set of valid strengthening inequalities for the IP (\ref{IP:first}). These inequalities are used to  strengthen the minimum cut constraints. These inequalities are needed because without them the LP has an unbounded integrality gap even if all jobs arrive at the same time on a single machine \cite{CheungS11}. Other natural LP relaxations, such as a time indexed LP also have an unbounded gap even on a single machine. 

To derive the inequalities, notice that any given IP solution can be related to a corresponding flow graph.   Each job node $a_i$ has an edge to a node $b_t$ in the right bipartition with capacity $x_{i,t}$.  These will all be $1$ for $t$ smaller than $i$'s completion time and $0$ for larger $t$.  The only edges that depend on the completion times for jobs are those between the job nodes and the nodes representing time steps and, therefore, this completes the graph description.   We will call the resulting graph $G$ a \textbf{scheduling flow graph}.

Fix a solution $x$ to the IP. Let $C_j$ be the completion time of job $j$ in this solution; that is, the latest time $t$ where $x_{j,t}=1$. Let $G_x$ be the corresponding scheduling flow graph. In $G_x$ the job node $a_j$ has edges of capacity $1$ to a contiguous set of nodes on the right partition from time one  until its completion time. That is, to all nodes $b_1, b_2, \ldots , b_{C_j}$.  This  property is what will be leveraged to derive strengthening inequalities.  To do so, first some basic facts about these graphs are established.  The first relates cuts in the graph to the integer program and this follows by definition of the constraints (\ref{constraint:weak}).

\begin{fact}
Let $x$ be a solution to the IP and let $G_x$ be the resulting scheduling flow graph. If the minimum cut in $G_x$ is at least $\sum_{i \in J} p_i$ then the set of constraints (\ref{constraint:weak}) are satisfied in $x$.
\end{fact}

The previous fact allows us to focus on strengthening  the inequalities that correspond to  a minimum cut in $G_x$. Next consider two facts about minimum cuts in any scheduling flow graph.
  
  \begin{fact}
  \label{fact:pedges}
Consider any scheduling flow graph $G$ and any cut $\{s\} \cup A' \cup (B\setminus B')$ where    $A' \subseteq A$ and $B'\subseteq B$.   If this is a \emph{minimum cut} then  any node $a_j \in A'$ has at most $p_j$ edges to nodes in $B'$ with positive capacity (i.e. equal to 1). 
  \end{fact}
  \begin{proof}
  For the sake of contradiction say that the node $a_j$ has more than $p_j$ edges to nodes in $B'$ with positive capacity.  Each of the edges is cut, all have capacity $1$ and their total capacity is at least $p_j+1$.  By moving $a_j$ to $A \setminus A'$ all of these edges will no longer be cut.  The only new edge cut is from $s$ to $a_j$ of capacity $p_j$.  This strictly decreases the total cut value, contradicting the definition of  a minimum cut. 
  \end{proof}
  
    \begin{fact}
    \label{fact:bigjob}
Consider any scheduling flow graph $G$ and any minimum cut $\{s\} \cup A' \cup (B\setminus B')$ where    $A' \subseteq A$ and $B'\subseteq B$.   Fix a job $j$ with sufficiently large processing time, $p_j > \sum_{i \in J} p_i - m | B\setminus B'|$.  If $a_j$ has at least $ \sum_{i \in J} p_i - m | B\setminus B'|$  edges to nodes in $ B'$ with positive value (i.e. equal to 1), then the minimum cut value is at least $\sum_{i \in J} p_i$.
   \end{fact}
\begin{proof}
Since the cut is a minimum cut, it suffices to bound the value of this cut by $\sum_{i \in J} p_j$ to establish the fact.  Consider the node $a_j$. If $a_j \notin A'$ then the edge from the source to $a_j$ is cut and its capacity is $p_j$. If $a_j \in A'$ then $ \sum_{i \in J} p_i - m | B\setminus B'|$ edges from $a_j$ to nodes in $B'$ are cut and they all have capacity $1$.  Given that $p_j$ is sufficiently large, in either case the total capacity of edges adjacent to $a_j$ that are cut is at least $ \sum_{i \in J} p_i - m | B\setminus B'|$ 

Additionally, notice that the total capacity of edges from time step nodes to the sink that are cut is exactly $m | B\setminus B'|$.  Together this shows the cut value is at least $\sum_{i \in J} p_i$. 
\end{proof}
  
  \medskip
\noindent \textbf{First Set of Strong Inequalities:}  Fix any solution $x$ and a minimum cut  $\{s\} \cup A' \cup (B\setminus B')$ where  $B'\subseteq B$ and $A'\subseteq A$ in the graph $G_x$ corresponding to $x$.  Let $T'$ be the time steps corresponding to nodes in $ B'$ and let $J'$ be the jobs corresponding to nodes in $A'$. The two prior facts will be used to give inequalities  that strengthen the constraint in (\ref{constraint:weak}) for $T'$ and $J'$.

  The constraint for $T'$ and $J'$ is satisfied only if the following is satisfied. For a job $j$ let $E(T',j)$ be the set of up to $\min\{p_j, \sum_{i \in J} p_i - m | T\setminus T'| \}$ \textbf{earliest} time steps in $ T'$. 
 
 \vspace{-.3cm}
  \begin{small}
\begin{eqnarray} \sum_{j \in J \setminus J'}p_j + \sum_{j \in J'} \sum_{t \in  E(T',j)} x_{j,t} + \sum_{t \in T\setminus T'} m \geq \sum_{j \in J} p_j \label{eqn:inter}
\end{eqnarray}
\end{small} 
\vspace{-.3cm}

The  change is that the summation in the second term only considers time steps in $E(T',j)$ for each job $j$. For some solution, this could potentially reduce the left hand side of the constraint. To see why this is a valid replacement for the constraint consider the following.  Fix job $j$ and let $t^*_j$ be the latest time in $E(T',j)$. Only the time steps in $E(T',j)$ need to be considered in the summation because one of the following holds.  

\begin{compactitem}
\item If $j$ is completed at or before time $t_j^*$ (i.e. $x_{j,t} = 0$ for $t > t_j^*$), then $j$ contributes the same amount to the left hand side of (\ref{eqn:inter}) as  (\ref{constraint:weak}).  Due to this, in the following bullets assume $j$ is completed after $t_j^*$ (i.e. $x_{j,t} = 1$ for $t \leq t_j^*$).
\item If there are $p_j$ time steps at or before $t^*_j$ in $T'$ and $j$ is completed after $t_j^*$, it must be the case that $j \notin J'$ by Fact~\ref{fact:pedges}.  Hence $j$ is not included in the outer summation.
\item Say the prior case doesn't hold. If $ \sum_{i \in J} p_i - m | T\setminus T'|$ time steps are in $T'$ at or before $t^*_j$ and $j$ is completed after $t_j^*$ then the proof of Fact~\ref{fact:bigjob} implies the constraint is satisfied by only considering these edges in $E(T',j)$.  Note that if the previous case did not hold, then in this case we may assume $p_j > \sum_{i \in J} p_i - m | T\setminus T'|$ because there are at least $\sum_{i \in J} p_i - m | T\setminus T'|$ times in $T'$ at or before $t^*_j$, but less than $p_j$. 
\end{compactitem}

We now simplify these constraints.  Given that $|E(T',j)| \leq p_j$ for all jobs $j$ and that $x_{j,t}\leq 1$, the left hand side of the above constraint can be bounded as follows.   The modification is the first term is dropped, but now the second summation is over all jobs.

\vspace{-.3cm}
  \begin{small}
\begin{eqnarray*}
&&\sum_{j \in J \setminus J'}p_j + \sum_{j \in J'} \sum_{t \in  E(T',j)} x_{j,t} + \sum_{t \in T\setminus T'} m
\geq \sum_{j \in J \setminus J'}|E(T',j)| + \sum_{j \in J'} \sum_{t \in  E(T',j)} x_{j,t} + \sum_{t \in T\setminus T'} m\\
&\geq& \sum_{j \in J} \sum_{t \in E(T',j)} x_{j,t} + \sum_{t \in T\setminus T'} m   \\
\end{eqnarray*}
\end{small} 
\vspace{-.7cm}

This shows that the above constraint for $J'$ and $T'$ in (\ref{eqn:inter}) is satisfied  only if the  constraint in (\ref{eqn:inter}) for $J$ and $T'$ is satisfied.  This will allow us to define inequalities that do not iterate over all subsets of jobs, but only subsets of time steps.   The set of constraints (\ref{constraint:weak}) in the integer program can be replaced with the following stronger set of inequalities.  These inequalities are stronger because we pruned the set of $x$ variables that can contribute to the left hand side of the constraint in the first summation. 

\vspace{-.5cm}
  \begin{small}
\begin{eqnarray}\sum_{j \in J} \sum_{t \in E(T',j)} x_{j,t} + \sum_{t \in T\setminus T'} m \geq \sum_{j \in J} p_j  \;\;\;\;  \forall T' \subseteq T  \label{eqn:firstcut}
\end{eqnarray}
\end{small} 
\vspace{-.5cm}

  \medskip
\noindent \textbf{Even Stronger Inequalities:} To further strengthen these inequalities, ideas similar to knapsack cover inequalities are used.  Let $D$ be a \textbf{vector} where $D_j$ is a time corresponding to job $j$.   Intuitively, $D_j$ is a lower bound on the completion time for  job $j$. The constraints below say that even if all jobs $j$ are set to have completion times \textbf{at least} $D_j$  then the constraints should still be satisfied.   

Fix $T' \subseteq T$ and a vector $D$ of completion times for every job.  Consider the constraint for $T'$ in (\ref{eqn:firstcut}).  If job $j$ is given a completion time of at least $D_j$ then  $j$ will contribute at least $\sum_{t \in  [0, D_j] \cap E(T',j) }1$ to the left side of the inequality.  Let $V(T',D) = \sum_{i\in J } p_i - \sum_{t \in T \setminus T'} m - \sum_{j \in J}\sum_{t \in  [0, D_j] \cap E(T',j) }1$.  This idea gives rise to the following set of constraints for each pair $T'$ and $D$ where $V(T',D)$ is positive.  

\vspace{-.3cm}
\begin{small}
$$\sum_{j \in J} \sum_{t \in E(T',j), t > D_j} x_{j,t} \geq V(T',D)  \;\;\;\;  \forall T' \subseteq T, \forall D, V(T',D)  >0$$
\end{small}
\vspace{-.3cm}

These new constraints are not tighter than those in (\ref{eqn:firstcut}). To make them tighter, consider the following.   

  Let $E(T',D,j)$ be the \textbf{earliest}  $V(T',D)$ time steps  in $E(T',j)$ later than $D_j$. The new constraints are as follows.  These are the \textbf{job cover inequalities} and are the same as constraints~\ref{constraint:nospeed1} givin in the paper. 

\vspace{-.4cm}
\begin{eqnarray*}
\sum_{j \in J} \sum_{t \in E(T',D,j)} x_{j,t} \geq V(T',D)  \;\;\;\;  \forall T' \subseteq T, \forall D, V(T',D)  >0  \qquad (\ref{constraint:nospeed1}) 
\end{eqnarray*}
\vspace{-.4cm}

Knowing that $E(T',D,j) \subseteq \{ t \;| \; t \in E(T',j) \mbox{ and } t > D_j\} $ for any $D$, these inequalities are tighter because the left hand side only considers times in $E(T',D,j)$.  These constraints can be shown to be valid using the following observations.  Consider any job $j$.  We compare the summation $\sum_{t \in E(T',j)} x_{j,t}$  to the summation  $\sum_{t \in E(T',D,j)} x_{j,t} $ and argue that changing the first to the second will not effect the validity of the constraint. In particular, this change can only decrease the left hand side of the inequality, so we need to establish that the constraint remains satisfied after the change if it was originally satisfied.  In the following, let $t^*_j$ denote the last time in $E(T',D,j)$.

\begin{compactitem}
\item If there are at most $V(T',D)$ time steps  in $E(T',j)$ later than $D_j$ then these summations are the same.
\item   If there are more than $V(T',D)$ time steps  in $E(T',j)$ later than $D_j$ and job $j$ completes at or \emph{before} time $t^*_j$ then the two summations are the same.  This is because $x_{j,t} =0$ for all times $t \in E(T',j) \setminus E(T',D,j)$ because $j$ completes by these times. 
\item   If there are more than $V(T',D)$ time steps  in $E(T',j)$ later than $D_j$ and job $j$ completes \emph{after} time $t^*_j$ then  $\sum_{t \in E(T',D,j)} x_{j,t}  \geq V(T',D) $ and the constraint in  (\ref{constraint:nospeed1}) for $T'$ and $D$ is satisfied.  This is because $x_{j,t} =1$ for all times $t \in E(T',D,j)$ and there are at least $V(T',D)$ such times.  
\end{compactitem}



\medskip 
\noindent \textbf{Reducing the Number of Constraints:} In Appendix \ref{sec:continuousT} it is established that any \emph{integer} solution satisfies all the constraints in (\ref{constraint:nospeed1}) if and only if the following are satisfied. It can be established that any \emph{integer} solution satisfies all the constraints in (\ref{constraint:nospeed1}) if and only if the following are satisfied.  These are the same as the constraints~ (\ref{constraint:nospeed}) presented in the main body of the paper. 

\vspace{-.4cm}
\begin{eqnarray*}
\sum_{j \in J} \sum_{t \in E(T',D,j)} x_{j,t} \geq V(T',D)  \;\;\;\;   \forall b \in [0,\infty], T' =[b,\infty], \forall D, V(T',D)  >0 \qquad (\ref{constraint:nospeed})  
\end{eqnarray*} 

 These constraints state that if the subset of constraints in  (\ref{constraint:nospeed1}) are satisfied for any $D$ and all sets $T'$  that consist of a continuous set of time steps from some  time $b$ to time $\infty$ then all of the constraints in  (\ref{constraint:nospeed1})  are satisfied (for any $T'$).  Due to this, we will only need to use the constraints in (\ref{constraint:nospeed}) that restricts the sets $T'$.

The constraints (\ref{constraint:weak}) in the  IP are replaced by the constraints in (\ref{constraint:nospeed}).   Throughout the paper these constraints are discussed and it is said that a fixed constraint is defined by the set $T' = [b,\infty]$ and a vector $D$.

\medskip \noindent \textbf{Note on Solving  the LP:}  The IP is relaxed to a LP.  The LP is solved and then subsequently rounded to an integer solution.  There are an exponential number of constraints.  To solve the LP, the ellipsoid method is used.  The author does not know of an efficient separation oracle for the set of constraints in (\ref{constraint:nospeed1}).   The reduced set of constraints in (\ref{constraint:nospeed}) are the only constraints needed for the analysis.  For this set of constraints, an efficient dynamic programming algorithm can be used as a separation oracle.  See Appendix~\ref{sec:separation}.


\section{Reducing the Number of Constraints}
\label{sec:continuousT}

This section shows that if all the constraints in (\ref{constraint:nospeed})  are satisfied then the superset of constraints in (\ref{constraint:nospeed1}) are satisfied.  In particular, if the constraints where $T'$ is a contiguous set of time steps from some time $b$ until time $\infty$ for all values of $b$ are satisfied then the constraints are satisfied for all possible choices of $T'$.

 To begin the proof, the following lemma relates how the constraints change as time steps are added to $T'$. In particular $\sum_{j \in J} \sum_{t \in E(T',j)  }  x_{j,t}$ decreases if a time step is removed from $T'$ and is replaced with a later time step.   This follows because the $x$ variables decrease as $t$ becomes larger by constraint $(\ref{constraint:greater})$.  Note that this summation counts times in $E(T',j)$ and \emph{not} $E(T',D',j)$.
 
 \begin{lemma}
 \label{lem:swaptimesoffline}
 Fix any (possibly infeasible) solution $x'$ to the LP  that satisfies $x'_{j,t} \geq x'_{j,t+1}$ for all jobs $j$ and times $t$.   Fix any set of time steps $T'$.  For any $t \in T'$  and any  time $t' > t$ where $t' \notin T'$  it is the case that $\sum_{j \in J} \sum_{t \in E(T',j) }  x'_{j,t} \geq \sum_{j \in J} \sum_{t \in E(T' \setminus \{t\} \cup \{t'\},j)}  x'_{j,t}$.  
 \end{lemma}
 
 \begin{proof}
 By definition $E(T',j)$ contains the earliest $\min\{p_j, \sum_{i\in J} p_i - m|T\setminus T'| \}$ times steps in $T'$.  The definition implies that $E(T' \setminus \{t\} \cup \{t'\},j)$ contains exactly as many time steps as $E(T',j)$.  Further, since $t' \geq t$ it must the the case that  either $E(T' \setminus \{t\} \cup \{t'\},j) = E(T',j)$ or $t \in E(T',j)$ and there is a time step $t'' \in E(T' \setminus \{t\} \cup \{t'\},j) \setminus E(T',j)$ that is only later than $t$.  Knowing that  that $x'_{j,t} \geq x'_{j,t+1}$ for all $j$ and $t$ the lemma follows.

 \end{proof}
 
Next it is established that for integer solutions to the LP it is sufficient to focus on showing a weaker constraint than that used in the LP is satisfied.

\begin{lemma}
\label{lem:simpleeqoffline}
Consider any integer solution $x'$ the LP that satisfies $x'_{j,t} \geq x'_{j,t+1}$ for all $j,t$.  Fix any set of time steps $T'$. The inequality $\sum_{j \in J} \sum_{t \in E(T',j)} x'_{j,t} + \sum_{t \in T\setminus T'} m \geq \sum_{j \in J} p_j$ is satisfied if and only if  constraints for $T'$ and all $D$ are satisfied by $x'$ in (\ref{constraint:nospeed1}).
\end{lemma}
\begin{proof}

The first inequality is equivalent to  (\ref{constraint:nospeed1})  when $D_j = 0$ for all $D_j$.  Due to this, it suffices to show that the if there are sets $T'$ and $D$ where the constraint  (\ref{constraint:nospeed1}) is unsatisfied then $\sum_{j \in J} \sum_{t \in E(T',j)} x'_{j,t} + \sum_{t \in T\setminus T'} m < \sum_{j \in J} p_j$.  Fix such a set $T'$ and a vector $D$.    By assumption, it is the case that the following holds.

$$\sum_{j \in J} \sum_{t \in E(T',D,j)} x'_{j,t} < V(T',D)  = \sum_{i\in J } p_i - \sum_{t \in T \setminus T'} m - \sum_{j \in J}\sum_{t \in  [0, D_j] \cap E(T',j) }1 $$

The definition of $E(T',D,j)$ and $E(T,j)$ imply that for any $j$ it is the case that $ \sum_{t \in E(T',D,j)} x'_{j,t} +  \sum_{j \in J}\sum_{t \in  [0, D_j] \cap E(T',j) }1 \geq  \sum_{t \in E(T',j)} x'_{j,t}$.  Equality holds if $D_j$ is smaller than the last time $t$ where $x_{j,t} =1$. This and the above gives the following, proving the lemma

$$\sum_{j \in J} \sum_{t \in E(T',j)} x'_{j,t} <  \sum_{i\in J } p_i - \sum_{t \in T \setminus T'} m $$

\end{proof}

Using the previous two lemmas, it can  be shown that if there is a constraint in (\ref{constraint:nospeed1}) for an arbitrary $T'$ and $D$ unsatisfied by an integer solution $x^*$ then there is a constraint unsatisfied  in (\ref{constraint:nospeed}) where $T'$ is of the form $T' =[b,\infty]$ for some $b$. This is established in the next lemma.

\begin{lemma}
\label{lem:continuousoffline}
Let $x^*$ be an integer solution satisfying $x^*_{j,t} \geq x^*_{j,t+1}$ for all $j,t$.  If there is a constraint in (\ref{constraint:nospeed1}) unsatisfied by $x^*$ for some $T'$ and $D$ then there exists a constraint in (\ref{constraint:nospeed}) unsatisfied by  $x^*$ for some $T''$ and $D''$ where $T'' = [b,\infty]$  for some $b$ and $D''$ contains $D''_j = 0$ for all jobs $j$.
\end{lemma}
\begin{proof}
Let $T'$ and $D$ be an unsatisfied constraint in (\ref{constraint:nospeed1}).  Say that $T'$ is not of the form $[b, \infty]$ for some $b$.  Then there exists a time $t' \in T'$ where there is both an earlier time than $t'$ in $T\setminus T'$ and a later time than $t'$ in $T\setminus T'$.  Fix $T'$  and $D$ to be the sets that minimizes the number of such  time steps over all $T'$ and $D$. That is, minimizes $|\{t \; | \; t \in T', \exists t',t''\in T\setminus T' \mbox{ s.t. } t' <t< t'' \}|$. Let $D''$ be set so that $D''_j = 0$ for all jobs $j$. Since this constraint is unsatisfied, Lemma~\ref{lem:simpleeqoffline} implies that the following holds.

$$\sum_{j \in J} \sum_{t \in E(T',j)} x^*_{j,t} + \sum_{t \in T\setminus T'} m < \sum_{j \in J} p_j$$
 
Let  $t' \in T'$  be a time where there is both an earlier time than $t'$ in $T\setminus T'$ and a later time that $t'$ in $T\setminus T'$.   Let $t''$ be the earliest time in $T\setminus T'$.  Lemma~\ref{lem:swaptimesoffline} implies that $\sum_{j \in J} \sum_{t \in E(T' \setminus \{t'\}\cup \{t''\},j)} x^*_{j,t}  \leq \sum_{j \in J} \sum_{t \in E(T',j)} x^*_{j,t}$. Thus,  it is the case that the following holds.

$$\sum_{j \in J} \sum_{t \in E(T'\setminus \{t'\}\cup \{t''\},j)} x^*_{j,t} + \sum_{t \in T\setminus (T'\setminus \{t'\}\cup \{t''\})} m < \sum_{j \in J} p_j$$

However, this precisely says that the constraint for $(T' \setminus \{t''\}\cup \{t'\})$ and $D''$ is unsatisfied.  This is a contradiction because   $T'$ and $D$ where chosen to minimize the number of time steps in $T'$  that have both a smaller and greater time in $T\setminus T' $.  The set  $(T' \setminus \{t'\}\cup \{t''\})$  has one less such time than $T'$.
 
\end{proof}

\section{Separation Oracle}
\label{sec:separation}

The algorithm uses the ellipsoid method to solve the linear program.  There are an exponential number of constraints in the set of constraints in (\ref{constraint:nospeed}).  For this set of constraints there is an efficient separation oracle.

The separation oracle works as follows. First fix a set $T' = [b,\infty]$.  Since $T'$ is a continuous set of time-steps there are only $\texttt{poly}(P, n)$ possible sets and we try all such sets.  Fix a solution $x$.  We wish to determine if there is a constraint defined by $T'$ and some $D'$ that is unsatisfied.  That is,

$$\sum_{j \in J} \sum_{t \in E(T',D',j)} x_{j,t} < V(T',D').$$


 The algorithm first guesses $G$ the value of $V(T',D')$.  There are at most $\texttt{poly}(\Delta, n)$ possible values that the algorithm needs to enumerate over.  Now the algorithm uses dynamic programming.  Order the jobs arbitrarily from $1$ to $n$.  The algorithm recursively computes $M_{G}[i,\ell]$ which stores the minimum value of $\sum_{j=1}^i \sum_{t \in E(T',D,j)} x_{j,t}$ over all $D$ such that (1) $\ell = \sum_{j=1}^i \sum_{t \in [0,D_j] \cap E(T',j)} 1$ and (2) it must be case that $V(T',D)= G$.  There is a constraint unsatisfied for $T'$ and some $D$ if  $M_G[n,\ell] < G$ for $\ell = \sum_{i \in J} p_i - \sum_{t \in T\setminus T'} m - G$. 

The value of  $M_{G}[i,\ell]$ can be computed recursively as follows.   For ease of notation let $T_j$ denote the $j$ earliest time steps of $T'$.  Set $M_{G}[i,\ell] = \min_{1 \leq k \leq \ell}  \{ M[i-1,\ell -k] + \sum_{t \in (E(T',i) \setminus T_k) \cap T_{k+G})} x_{i,t}\} $.  For the base case, set $M_G[1, \ell] = \infty$ if $\ell > |E(T',1)|$ and otherwise $\sum_{t \in (E(T',1) \setminus T_\ell) \cap T_{\ell+G})} x_{i,t}  $.

\section{Removing the Dependence on $P$ in the Running Time}
\label{sec:largenumebr}

In this section, we show how to remove the dependence on $P$ in the running time.  To do so, the bipartite flow graph described in Section~\ref{sec:deadline} is modified to make its size polynomial in $n$ and $\log W$  where $W = \max_{j,t} g_j(t)$. For simplicity, it is assumed that $g_j$ returns integer values for all $j$.  For each job $j$ let $T_j$ be the set of time steps $t$ where $g_j(t) = 2^k$ for some $k$.  These are the time steps where the cost of job $j$ doubles.  Let $\mathcal{T} = \cup_jT_j$.  Note that $|\mathcal{T}|$ is  $O(n \log W)$.  

 Fix any schedule and let $C_j$ denote the completion time of job $j$ in the schedule. Sort the times in $\mathcal{T}$ so that $t_1 < t_2 < t_3 <\ldots t_{|\mathcal{T}'|}$.   Let $\ell_i = t_{i} - t_{i-1}$ for all $i > 1$. Let $G$ be a bipartite flow network as described  in Section~\ref{sec:deadline}, which is modified as follows. Remove all nodes in the right bipartition of $G$ and their adjacent edges.  Add a node $b_i$ on the right hand side for each time $t_i$  in $\mathcal{T}$.    Let $C_j$ be the given completion time for job $j$. Round $C_j$ to the nearest later time in $\mathcal{T}$.  Note that if the original solution is feasible, this new one is feasible.  Further the definition of times in $\mathcal{T}$ implies than at most a factor of $2$ is lost in the approximation ratio due to this rounding.  There is an edge from the node corresponding to job $j$,  $a_j$, to $b_i$ if $t_i \in [0,C_j]$. The capacity of the edge is $\ell_i$.  Intuitively, this capacity ensures that  a job can only run for $\ell_i$ time steps during $[t_{i-1}, t_i)$.  This is the maximum feasible amount a job can run during the interval.  The node $b_j$ has an edge to the sink of capacity $m\ell_i$.  This is the maximum amount of work all machines can process during $[t_{i-1}, t_i)$.
 
 The next lemma establishes that if there is a feasible flow in this graph of value $\sum_j p_j$ then there is a corresponding feasible schedule.
 
 \begin{lemma}
 If there is a feasible flow in $G$ of value $\sum_j p_j$ then there is a feasible corresponding schedule. 
 \end{lemma}
 
 \begin{proof}
 
 Let $y_{j,i}$ be the amount of flow that job $j$ sends through node $b_i$. To prove the lemma, it suffices to ensure that for every interval $[t_i,t_{i+1})$  there is a way to  feasibly schedule   all jobs $j$ for $y_{j,i}$ units during $[t_i,t_{i+1})$. Fix an interval $[t_i,t_{i+1})$.  Order the machines arbitrarily $M_1, M_2, \ldots M_m$.  Order the jobs arbitrarily $1, 2, \ldots n$.  Consider the jobs in order.  When scheduling job $i$ assign it to the lowest indexed machine and earliest time greater than $t_i$ that the machine is free. Machines are filled until time  $t_{i+1}$.  When considering job $i$ if the  machine has less than $y_{j,i}$  time steps free in $[t_{i-1}, t_i)$, schedule as many as possible until $t_i$ and move to the next machine.  
 
 Notice that this schedule ensures that each machine  $M_i$ is completely busy during $[t_i,t_{i+1})$ before the a job is scheduled on machine $M_{i+1}$. Knowing this, the only way for a job not to be feasibly scheduled is for two reasons.  The first is because all machines are completely busy.  However, this implies that $\sum_{j} y_{j,i} > m \ell_i$, a contradiction to the capacity of the edge from $b_j$ to the sink. The other is that a job is scheduled on two machines at the same time.  However, this cannot occur by definition of the algorithm since the only way for it to be scheduled on two machines at the same time is if $y_{j,i} > \ell_i$ and we know that $y_{j,i} \leq \ell_i$ by the capacity of the  edge from $a_j$ to $b_i$.
 
 \end{proof}
 
 Using this new bipartite flow graph, a new linear program is constructed.  The differences are the following.  There are only variables $x_{j,t}$ for times in $\mathcal{T}$.  The constraints in (\ref{constraint:nospeed}) are modified.  In particular, the constraint is defined for all sets $\mathcal{T}'= [t, \infty] $ for $t \in \mathcal{T}$ and all $D$.  The vectors $D$ are kept the same.  It is not hard to see that each fixed constraint is polynomial time computable. 
 
 With the LP established, the rounding and analysis are essentially identical.

\section{Omitted Proofs}
\label{sec:omitted}

\begin{proofof}[Lemma \ref{lem:feasible}]

All constraints in the LP are trivially satisfied except those in (\ref{constraint:nospeed}).  Fix any  constraint $T',D'$ where $D'_j \geq C_{j,\alpha}$ for all jobs $j$.  Note that we need to only consider such constraints because all constraints are satisfied if these are since $\tilde{x}^*_{j,t}=1$ for $t \leq C_{j,\alpha}$ and the definition of the constraints.  

We now establish that  $\sum_{j \in L_{T',D'} \cup M_{T',D'}} \sum_{t \in E(T',D',j)} x'_{j,t} \leq \frac{4}{10} V(T',D')$. Let $j^*$ be the job in $L_{T',D'}$ such that $D'_{j^*}$ is maximized. It is the case that $\sum_{j \in L_{T',D'} \setminus \{j^*\}} \sum_{t \in E(T',D',j)} x'_{j,t} < \frac{1}{10} V(T',D')$ by definition of $L_{T',D'}$. Further, $|E(T',D',j^*)| \leq V(T',D')$ by definition. Knowing that $x'_{j,t} \leq c$ for $t > C_{j,\alpha}$ and times  in $E(T',D',j)$ are later than $C_{j,\alpha}$ it is the case that, $  \sum_{t \in E(T',D',j^*)} x'_{j^*,t} \leq c|E(T',D',j^*)| \leq cV(T',D')$.  Thus,  $\sum_{j \in L_{T',D'} } \sum_{t \in E(T',D',j)} x'_{j,t} \leq \frac{2}{10} V(T',D')$ for $c \leq 1/10$.  

An almost identical line of reasoning can be used  to show that $\sum_{j \in M_{T',D'} } \sum_{t \in E(T',D',j)} x'_{j,t} \leq \frac{2}{10} V(T',D')$ for $c \leq 1/10$.   By summing these two expressions we have that $\sum_{j \in L_{T',D'} \cup M_{T',D'}} \sum_{t \in E(T',D',j)} x'_{j,t} \leq \frac{4}{10} V(T',D')$ as desired.   

We know that  $\sum_{j \in J} \sum_{t \in E(T',D',j)} x'_{j,t} \geq  V(T',D')$ since $x'$ is a feasible solution.  Therefore, $\sum_{j \in \cJ_{T',D'}} \sum_{t \in E(T',D',j)} x'_{j,t} \geq \frac{6}{10}V(T',D')$ because $\cJ_{T',D'}$ includes all jobs contributing a positive amount to the summation besides those in $L_{T',D'} \cup M_{T',D'}$ and combining this with the above observation.

 We further know that for all jobs $j$ in $\cJ_{T',D'}$ it is the case that $\tilde{x}^*_{j,t} = 10  x'_{j,t}$ for $t \geq C_{j,\alpha}$.  Thus, it is the case that  $\sum_{j \in J} \sum_{t \in E(T',D',j)} \tilde{x}^*_{j,t} \geq \sum_{j \in \cJ_{T',D'}} \sum_{t \in E(T',D',j)} \tilde{x}^*_{j,t} \geq \sum_{j \in \cJ_{T',D'}} \sum_{t \in E(T',D',j)} 10x'_{j,t} \geq 10  \frac{6}{10} V(T',D') = 6V(T',D')$ and the constraint is satisfied. 
\end{proofof}

\begin{proofof}[Lemma \ref{lem:expectedcost}]
  Let $C_{j,c}$ be the  latest time $t$ where $x'_{j,t} \geq c$. Consider a fixed job $j$.  The probability that time $C_{j,\alpha} =t$ is exactly $\frac{1}{c}(x'_{j,t} - x'_{j,t+1})$ for $t > C_{j,c}$. The solution $x^*$'s cost for job $j$ is $g_j(C_{j,\alpha})$.  By the linearity of expectation, the expected objective of $x^*$ is:

$$\sum_{j\in  J} \left(  \sum_{t\geq C_{j,c}} \frac{1}{c}(x'_{j,t} - x'_{j,t+1}) g_j(t)  \right ).$$

Consider the total cost of the solution $x'$, $ \sum_{j\in J} \sum_{t } x'_{j,t}(g_j(t) - g_j(t-1))$.  Note that this includes both the integral and fractional cost of $x'$'s objective.   This is equal to:

$$\sum_{j\in  J}  \left(  \sum_{t} (x'_{j,t} - x'_{j,t+1}) g_j(t)  \right )$$

This is at least the following knowing that $x'_{j,t} \geq c$ for all $0\leq t \leq C_{j,c}$ and $x'_{j,t} = 1$ for all $0\leq t \leq \beta'_j$.

$$\sum_{j\in  J}  \left( (1-c)g(\beta'_j) +  \sum_{t \geq C_{j,c}} (x'_{j,t} - x'_{j,t+1}) g_j(t)  \right ) \geq \sum_{j\in  J}   \sum_{t \geq C_{j,c}} (x'_{j,t} - x'_{j,t+1}) g_j(t)  $$

This is at most a $c$ factor smaller than the expected cost of the solution $x^*$.   Due to this, it suffices to prove that this expression is smaller than $\lpfrac$. This follows because $\lpfrac= \sum_{j\in J} \sum_{t > \beta'_j} x'_{j,t}(g_j(t) - g_j(t-1))\geq \sum_{j\in  J}   \sum_{t \geq \beta'_j} (x'_{j,t} - x'_{j,t+1}) g_j(t) \geq \sum_{j\in  J}   \sum_{t \geq C_{j,c}} (x'_{j,t} - x'_{j,t+1}) g_j(t)$ because $\beta'_j \leq C_{j,c}$.

\end{proofof}

\begin{proofof}[Proposition \ref{prop:runtime}]
We argue that the algorithm will terminate after all variables  in the LP are either integer or smaller than $1/(Pn)^2$.  At this point, the algorithm will terminate on the next iteration, satisfying all constraints with an integer solution.  To see this, fix a constraint in (\ref{constraint:nospeed}). There are at most $nP$ variables on the left hand side.  The total factional amount contributing if variables are less than $1/(Pn)^2$ is $1/(Pn)$.  Since this is less than $1$ and the right hand side is integer, for the constraint to be satisfied the integer potion of the left hand side must be larger than the right hand side. Thus the constrains remain satisfied even if all fractional variables are rounded to $0$. Lemma~\ref{lem:Finalcosts} gives that fractional objective is smaller than $1/(Pn)^2$ after $O(\log nP)$ iterations in expectation. After this number of iterations all fractional variables will be at most  $1/(Pn)^2$.
\end{proofof}

\begin{proofof}[Proposition~\ref{prop:allD}]
Fix any set  $T' = [t_{T'},\infty]$ and $D'$ as described. Say that $x^*$ satisfies the constraint for $T'$ and $D'$. We will show that  the constraint $T'$ and $D''$ is satisfied where $D''_j =0$ is in $D''$ for all jobs $j$.  Notice that if this constraint is satisfied then the constraints for $T'$ and any $D''$ is satisfied due to the validity of the constraints for integral solutions.  A formal proof of this fact can be found in Lemma~\ref{lem:continuousoffline} in the appendix.

Consider the constraint for $T'$ and $D'$.  Consider a job $j$ where $D'_j >0$.  Consider reducing $D'_j$ to $0$ in the constraint (\ref{constraint:nospeed}). Notice that $\sum_{t \in E(T',j), t > D'_j} x^*_{j,t} + |E(T',j) \cap [0,D'_j]| =  \sum_{t \in E(T',j), t > 0} x^*_{j,t}$ because $x^*_{j,t} = 1$ for all $t \leq D'_j$ by definition.   Thus, the left hand side of the constraint for $T'$ and $D'$ increases by at least much as the right hand side decreases.  This holds for all jobs. Thus, reducing $D'_j$ to $0$ for all jobs ensures the constraint remains satisfied for $T'$ and $D''$.
\end{proofof}

\begin{proofof}[Proposition~\ref{prop:during}]
By definition $E(T,D',j)$ contains the earliest $\min\{ p_j - \max\{0, D'_j - t_{T} \}, V(T,D')\}$ time steps after $D'_j$ for any set $T$ and vector $D'$. Note that $\min\{ p_j - \max\{0, D'_j - t_{T} \}, V(T,D')\} = \min\{ p_j , V(T,D')\}$ if $t_{T} > D'_j$.  Knowing that $V(T,D') \geq \frac{1}{2} V(T',D')$ and  $d _j \leq t_{T} < t_{T'}$ it is the case that $|E(T,D',j)| \geq \frac{1}{2}  |E(T',D',j)|$ by definition of the sets $E()$.  Knowing that $T' \subseteq T$,  $ |E(T,D',j)|$ contains only earlier times than $|E(T',D',j)|$.  Thus, constraint (\ref{constraint:greater}) stating that $x_{j,t} \geq x_{j,t+1}$ for all $j$ and $t$ implies that $\sum_{t \in E(T',D',j)} x_{j,t} \leq 2\sum_{t \in E(T,D',j)} x_{j,t}$.
\end{proofof}

\begin{proofof}[Proposition~\ref{prop:before}]
 By definition $E(T,D',j)$ contains the earliest $\min\{ p_j - \max\{0, D'_j - t_{T} \}, V(T,D')\}$ time steps after $D'_j$ for any sets $T$ and $D'$.   Knowing that $V(T,D') \geq \frac{1}{2} V(T',D')$ and $D'_j >t_{T} > t_{T'}$ it is the case that $|E(T,D',j)| \geq \frac{1}{2}  |E(T',D',j)|$ by definition of the sets $E()$.  Knowing that $E(T',D',j)$ and $E(T,D',j)$ only include times after $D'_j$,  it is the case that the earliest $ \frac{1}{2}  |E(T',D',j)|$ time in $E(T,D',j)$ are also the same earliest times in $E(T',D',j)$.  Thus, constraint (\ref{constraint:greater}) stating that $x_{j,t} \geq x_{j,t+1}$ for all $j$ and $t$ implies that $\sum_{t \in E(T',D',j)} x_{j,t} \leq 2\sum_{t \in E(T,D',j)} x_{j,t} $.

\end{proofof}

\begin{proofof}[Lemma~\ref{lem:offlinewhp}]

 Note that $E(T',D',j)$ is a subset of $[t_{T'}, t_{T'}+ \min\{V(T',D'),p_j\}]$ for any job $j$ by definition.    By definition of the algorithm, each job $j $ has $x^*_{j,t} = 1$ with probability $\frac{1}{c}x'_{j,t}$ if $x'_{j,t} < c$. Notice that $x'_{j,t}<c $ if $t \in  E(T',D',j) $ by definition of $D'$.     

We now show that $\sum_{j \in M_{T',D'}} \sum_{t \in E(T',D',j)} x^*_{j,t} < 10V(T',D')$ with probability at most $\frac{1}{ 2\log^{10} nP}$.  An identical proof shows that $\sum_{j \in L_{T',D'}} \sum_{t \in E(T',D',j)} x^*_{j,t} < 10V(T',D')$ with probability at most $\frac{1}{ 2\log^{10} nP}$ and hence this is omitted.  Once this is established, by the union bound the lemma follows. 

Let $X_j = \sum_{t \in E(T',D',j)} x^*_{j,t}$.  By definition of the algorithm, the random variables $X_j$ are independent.  Let $\mu = \ex[\sum_{j \in M_{T',D'}} X_j] = \ex[\frac{1}{c} \sum_{j \in M_{T',D'}} \sum_{t \in E(T',D',j)} x'_{j,t}] \geq \frac{1}{10c} V(T',D')$.  The last inequality follows by definition of $M_{T',D'}$.    Applying the concentration inequality in Theorem~\ref{thm:concentration-1} that can be found in Appendix~\ref{sec:concentration} and setting $\lambda = \frac{\mu}{2}\geq\frac{1}{20c}V(T',D')$ gives that
\begin{eqnarray*}
&&\Pr[\sum_{j \in M_{T',D'}} \sum_{t \in E(T',D',j)} x^*_{j,t} < 10V(T',D')] \leq \Pr[\sum_{j \in M_{T',D'}} \sum_{t \in E(T',D',j)} x^*_{j,t} < \frac{1}{20c}V(T',D')] \;\;\;\;  \\ 
&\leq& \Pr[\sum_{j \in M_{T',D'}} X_j< \frac{1}{2}\mu]  \leq \exp\Big( - \frac{\lambda^2}{2 \sum_{i \in M_{T',D'}} \ex(X_i^2)}\Big)\\
&\leq&\exp\Big( - \frac{\lambda^2}{2 V(T',D') \sum_{i \in M_{T',D'}} \ex(X_i)}\Big) \;\;\;\; \mbox{[$X_j \leq |E(T',D',j)| \leq V(T',D') $ for all $j$]}\\
&\leq&\exp\Big( - \frac{(\mu/2)^2}{2 V(T',D') \sum_{i \in M_{T',D'}} \ex(X_i)}\Big) \;\;\;\; \mbox{[$\lambda = \frac{\mu}{2}$]}\\
&=&\exp\Big( - \frac{\mu}{8 V(T',D')}\Big) \\
&=&\exp\Big( - \frac{1}{80c }\Big) 
\leq \frac{1}{2\log^{10} nP} \;\;\;\;\mbox{[$c \leq \frac{1}{1000\log\log nP}$, sufficiently large $n$ and $\mu \geq  \frac{1}{10c}V(T',D')$]}
\end{eqnarray*}
\end{proofof}

\begin{proofof}[Lemma~\ref{lem:befored}]

Fix any job $j^*$.     Let $D'$ be set such that $D'_{j}$ is the latest time  $t$ where $x'_{j,t} \geq c$ for all jobs $j$. The proof will establish that with probability greater $1-\frac{1}{\log^2 nP}$ it is the case that the constraints for $D'$ and any $T'=[t_{T'} ,\infty]$ with $t_{T'} < D'_{j^*}$ are satisfied by $x^*$.  Applying Proposition~\ref{prop:allD} this implies that $x^*$ satisfies the same allowing for any constraint $D''$, proving the lemma.

Geometrically group constraints based on the value of $V(T',D')$.  Let $\mathcal{C}_k$ contain the set $T' = [t_{T'} ,\infty]$ if  $j^* \in \mathcal{J}_{T',D'}$,  $2^k \leq V(T',D') < 2^{k+1}$ and $t_{T'} < D'_{j^*}$ for any integer $0 \leq k \leq \log nP$.  

Fix $k$ and the set $T' \in \mathcal{C}_k$ such that $t_{T'}$ is as early as possible.  Let $M_{T',D'}$ be as described in the algorithm definition. We will establish that if  $\sum_{j \in  M_{T',D'}} \sum_{t \in E(T',D',j)} x^*_{j,t} \geq 10V(T',D')$  then all constraints $V(T'',D')$ for any $T'' \in  \mathcal{C}_k$ are satisfied.  Once this is established, this will complete the proof as follows.  We apply  Lemma~\ref{lem:offlinewhp} stating that $\sum_{j \in  M_{T',D'}} \sum_{t \in E(T',D',j)} x^*_{j,t} \geq 10V(T',D')$ occurs with probability at least $1-\frac{1}{ \log^{10} nP}$.   By union bounding for all $\log nP$ values for $k$ the lemma follows.

Say that  $\sum_{j \in  M_{T',D'}} \sum_{t \in E(T',D',j)} x^*_{j,t} \geq 10V(T',D')$.  Consider any set $T'' \in \mathcal{C}_k$.   By definition of the  set  $M_{T',D'}$ it is the case that $D'_j \geq D'_{j^*}$. Hence, $D'_{j} \geq D'_{j^*} \geq t_{T''} \geq t_{T'} $. Thus, Proposition~\ref{prop:before} and the geometric grouping of constraints gives that $\sum_{j \in  M_{T',D'}} \sum_{t \in E(T'',D',j)} x^*_{j,t} \geq \frac{1}{2} \sum_{j \in  M_{T',D'}} \sum_{t \in E(T',D',j)} x^*_{j,t} \geq 5V(T',D') \geq V(T'',D')$.    Thus the constraint for $T''$ and $D'$ is satisfied, proving the lemma.

\end{proofof}

\section{Hardness of the GSP on Identical Machines without Preemption or Migration}
\label{sec:approxhard}

In this section, we establish that any algorithm has an unbounded approximation for GSP if preemption is not allowed or migration is not allowed.

\begin{proofof}[Theorem~\ref{thm:approxhard}]
Consider an instance of 3-partition with the parameter $B$ and items $a_1, a_2, \ldots a_n$ where $\sum_{i=1}^n a_i =\frac{Bn}{3}$.  We may assume $\frac{B}{4} < a_i < \frac{B}{2}$ for all $i$.  From this problem create an instance of the GSP.  Each item $i$ corresponds to a job with processing time $p_i = a_i$.  Let the number of machines be $m = n/3$.  For every job $i$ let $g_i(t) = 0$ for $t \leq B$ and $g_i(t) = W$ for $t > B$.  Here $W>0$ is an arbitrary parameter.  

In a feasible instance of the $3$-partition problem there exists sets $S_k$ such that $\sum_{a_i \in S_k} a_i = B$ and $|S_k|=3$ for all $1 \leq k \leq \frac{n}{3}$.  In such a instance, the optimal solution to the GSP has an objective of $0$ by scheduling jobs corresponding to items in $S_k$ on machines $k$.  Notice that the each job can be scheduled without preemption or migration.

Alternatively, say that the instance of $3$-partition is infeasible. If either preemption or migration is not allowed then each job must be scheduled on exactly one machine.  However, then there exists a machine where a job is not completed until after time $B$. Otherwise, each machine will have three jobs all scheduled before time $B$ giving a solution to $3$-partition.   This results in the optimal solution having an objective of at least $W$.

If an algorithm has a bounded approximation ratio for GSP then one can detect whether a $3$-partition instance is feasible or not.  This is because the ratio of the optimal solution in the corresponding GSP in a feasible versus infeasible instance of 3-partition is unbounded. 
\end{proofof}


\section{Concentration Inequality}
\label{sec:concentration}

The following theorem can be found in Chapter 2 in \cite{ChungL06}. 

\begin{theorem}
\label{thm:concentration-1}
	Let $X_1$, $X_2$, ..., $X_n$ be non-negative independent random variables, then we have the following bound for the sum $X = \sum_{i = 1}^{n} X_i$. 
		$$\Pr [ X \leq \ex X - \lambda ] \leq \exp\Big( - \frac{\lambda^2}{2 \sum_{i=1}^n \ex(X_i^2)}\Big)$$
\end{theorem}


%

\end{document}